\documentclass{article}
\usepackage{kr}

\usepackage{times}
\usepackage{soul}
\usepackage{url}
\usepackage[hidelinks]{hyperref}
\usepackage[utf8]{inputenc}
\usepackage[small]{caption}
\usepackage{graphicx}
\usepackage{amsmath}
\usepackage{amsthm}
\usepackage{booktabs}
\usepackage{algorithm2e}
\usepackage{tikz}
\usetikzlibrary{fadings}
\usetikzlibrary{shapes,decorations,positioning}
\usetikzlibrary{fadings,shapes,decorations,positioning,calc}
\usetikzlibrary{decorations.markings,decorations.pathreplacing}
\usetikzlibrary{shapes.geometric,automata,positioning}
\usetikzlibrary{shadows}\usetikzlibrary{calc}
\usetikzlibrary{decorations,decorations.pathmorphing,decorations.pathreplacing}
\usetikzlibrary{calc}
\usetikzlibrary{tikzmark}

\newtheorem{theorem}{Theorem}
\newtheorem{example}[theorem]{Example}
\newtheorem{definition}[theorem]{Definition}
\newtheorem{proposition}[theorem]{Proposition}
\newtheorem{lemma}[theorem]{Lemma}
\newtheorem{corollary}[theorem]{Corollary}

\title{Representation Theorems for Cumulative Propositional Dependence Logics}

\author{%
         Juha Kontinen$^{1}$\footnote{The authors are ordered alphabetically.} \and Arne Meier$^{2*}$\and Kai Sauerwald$^{3*}$
         \affiliations
         $^1$Department of Mathematics and Statistics, University of Helsinki, Helsinki, Finland\\
         $^2$Theoretical Computer Science, Leibniz University Hannover, Hannover, Germany\\
         $^3$Faculty of Mathematics and Computer Science, FernUniversität in Hagen, Hagen, Germany
     }

\usepackage{dlcobook}
\usepackage{xspace}
\usepackage{mathrsfs} 
\usepackage{mathtools}
\usepackage{amssymb,amsmath}
\usepackage{comment}
\usepackage{stmaryrd}
\usepackage{csquotes}
\usepackage[UKenglish]{babel}
\usepackage{enumerate}
\usepackage{microtype}
\usepackage{multicol}

\usepackage{pifont}%
\newcommand{\minOf}[2]{\ensuremath{\min(#1,#2)}}
\renewcommand{\tuple}[1]{\ensuremath{\langle{#1}\rangle}}
\newcommand{\statesOf}[1]{\ensuremath{\pmS(#1)}}
\newcommand{\modelsOf}[1]{\ensuremath{\llbracket #1\rrbracket}}

\newcommand{\pmC}{\ensuremath{\mathbb{C}}}
\newcommand{\pmM}{\ensuremath{\mathbb{M}}}

\newcommand{\pmS}{\ensuremath{\mathcal{S}}}
\newcommand{\pmL}{\ensuremath{\ell}}
\newcommand{\pmP}{\ensuremath{\prec}}
\newcommand{\pmR}{\ensuremath{\mathrel{R}}}

\DeclareMathOperator{\nmableitC}{\nmableit_{\mkern-2mu\pmC}}

\ifx\nmableit\undefined
\makeatletter
\ifx\centernot\undefined
\newcommand*{\centernot}{%
	\mathpalette\@centernot
}
\fi
\def\@centernot#1#2{%
	\mathrel{%
		\rlap{%
			\settowidth\dimen@{$\m@th#1{#2}$}%
			\kern.5\dimen@
			\settowidth\dimen@{$\m@th#1=$}%
			\kern-.5\dimen@
			$\m@th#1\not$%
		}%
		{#2}%
	}%
}
\makeatother
\DeclareRobustCommand\nmableitSymb{\mathrel{|}\joinrel\sim} %
\newcommand{\nmableit}{\nmableitSymb} %
\fi

 \newcommand{\allModels}[1]{\mathbb{A}_{#1}}
\newcommand{\allTeams}[1]{\mathbb{T}_{#1}}

\newif\ifhideproofs

\ifhideproofs
\usepackage{environ}
\NewEnviron{hide}{}

\fi

\newcommand{\PDL}{\logicFont{PDL}\xspace}

\renewcommand{\PL}{\ensuremath{\mathrm{PL}}\xspace}
\newcommand{\TPL}{\ensuremath{\logicFont{TPL}}\xspace}

\newcommand{\CPL}{\ensuremath{\logicFont{CPL}}\xspace}

\newcommand{\formulaOne}{\ensuremath{\varphi}}
\newcommand{\formulaTwo}{\ensuremath{\psi}}
\newcommand{\formulaThree}{\ensuremath{\gamma}}

\newcommand{\pmLogicPref}[1]{\ensuremath{{#1}\kern-0.5ex\raise0.95ex\hbox{\tiny\logicFont{pref}}}}
\newcommand{\TPLPref}{\pmLogicPref{\TPL}}
\newcommand{\CPLPref}{\pmLogicPref{\CPL}}
\newcommand{\PDLPref}{\pmLogicPref{\PDL}}

\newcommand{\pmLogicX}[2]{\ensuremath{{#1}\kern-0.5ex\raise0.95ex\hbox{\tiny{#2}}}}

\newcommand{\pmLogicCum}[1]{\ensuremath{{#1}\kern-0.5ex\raise0.95ex\hbox{\normalfont\tiny\logicFont{cuml}}}}
\newcommand{\TPLCum}{\pmLogicCum{\TPL}}
\newcommand{\CPLCum}{\pmLogicCum{\CPL}}
\newcommand{\PDLCum}{\pmLogicCum{\PDL}}

\newcommand{\pmLogicStrongCum}[1]{\ensuremath{{#1}\kern-0.5ex\raise0.95ex\hbox{\normalfont\tiny\logicFont{cuml}[\logicFont{str}]}}}
\newcommand{\TPLStrongCum}{\pmLogicStrongCum{\TPL}}
\newcommand{\CPLStrongCum}{\pmLogicStrongCum{\CPL}}
\newcommand{\PDLStrongCum}{\pmLogicStrongCum{\PDL}}

\newcommand{\pmLogicAsym}[1]{\ensuremath{{#1}\kern-0.5ex\raise0.95ex\hbox{\normalfont\tiny\logicFont{cuml}[\logicFont{as}]}}}
\newcommand{\TPLAsym}{\pmLogicAsym{\TPL}}

\newcommand{\pmLogicSystemP}[1]{\ensuremath{{#1}\kern-0.5ex\raise0.95ex\hbox{\tiny\textbf{\logicFont{p}}}}}
\newcommand{\pmLogicSystemC}[1]{\ensuremath{{#1}\kern-0.5ex\raise0.95ex\hbox{\tiny\textbf{\logicFont{c}}}}}
\newcommand{\TPLSystemP}{\pmLogicSystemP{\TPL}}
\newcommand{\CPLSystemP}{\pmLogicSystemP{\CPL}}
\newcommand{\PDLSystemP}{\pmLogicSystemP{\PDL}}
\newcommand{\TPLSystemC}{\pmLogicSystemC{\TPL}}
\newcommand{\CPLSystemC}{\pmLogicSystemC{\CPL}}
\newcommand{\PDLSystemC}{\pmLogicSystemC{\PDL}}

\newcommand{\SystemC}{\textit{System~C}}
\renewcommand{\pdl}{\ensuremath{\PL(=)}}
\renewcommand{\pdl}{\ensuremath{\PL(\logicFont{dep})}}

\newcommand{\ssLogic}{\ensuremath{\mathscr{L}}}
\newcommand{\ssSystem}{\ensuremath{\mathbb{S}}}
\newcommand{\ssFormulas}{\ensuremath{\mathcal{L}}}
\newcommand{\ssInt}{\ensuremath{\Omega}}

\usepackage{todonotes}

\newcommand{\ksCiteAY}[1]{\citeauthor{#1}~(\citeyear{#1})}

\newcommand{\MYParagraph}[1]{\par\smallskip\noindent\textbf{#1}}
\newcommand{\MYSubsection}[1]{\MYParagraph{#1.}}

\newcommand{\Th}{\mathrm{Th}}
\newcommand{\Cn}{\mathrm{Cn}}
\newcommand{\Cableit}{\mathrm{C}_{{\nmableit}}}

\makeatletter
\newcommand{\textlabelmarker}[1]{%
    \protected@edef\@currentlabel{#1}%
    \phantomsection%
}
\newcommand{\textlabel}[2]{%
    \textlabelmarker{#1}%
    #1\label{#2}%
}
\makeatother

\begin{document}
\thispagestyle{plain}

\maketitle

\begin{abstract}
    \vspace{-0.5em}
This paper establishes and proves representation theorems for cumulative propositional dependence logic and for cumulative propositional logic with team semantics. Cumulative logics are famously given by System~C. For propositional dependence logic, we show that System~C entailments are exactly captured by cumulative models from Kraus, Lehmann and Magidor. On the other hand, we show that entailment in cumulative propositional logics with team semantics is exactly captured by cumulative and asymmetric models.
For the latter, we also obtain equivalence with cumulative logics based on propositional logic with classical semantics.
The proofs will be useful for proving representation theorems for other cumulative logics without negation and material implication.
\end{abstract}

\vspace{-1.5em}
\section{Introduction}
\label{sec:introduction}
The ability to reason is one of the central features of intelligent agents and thus a major concern of artificial intelligence. 
In this paper, we study the fusion of non-monotonic reasoning with team-based reasoning.
Such a combination is interesting, as it allows reasoning in which one can express settings involving a plurality of objects (the team-based component) while also taking into account extra-logical information, such as plausibility, exceptionality, reliability, preference, or typicality (the non-monotonic component).
Notably, the combination of these approaches provides a setting that cannot be formalized by neither of the approaches alone.%

Less is known about how to construct non-monotonic entailment relations for team-based logics. Known approaches for non-monotonic propositional logics or non-monotonic first-order logics rely heavily on the properties of the underlying logics, such as, e.g., the availability of all Boolean connectives, the law of excluded middle, and presence of material implication. However, these properties are not (always) available in team-based logics.
 There are many approaches to non-monotonic logics that take inspiration from~\cite{KS_Brewka1997}, from which some are also generic, e.g., MAK models~\cite{KS_Makinson1989}, KLM-style reasoning~\cite{KS_KrausLehmannMagidor1990}, characterization logics~\cite{KS_BaumannStrass2025}, or approximation fixpoint theory~\cite{KS_DeneckerMarekTruszczynski2000}.

So far, there are only two pioneering works that consider a combination of team semantics and non-monotonic reasoning. First, the work by \citeauthor{JY23} (\citeyear{JY23}), which considers a non-monotonic team-based modal logic in the context of formal analysis of natural language. The second is by \citeauthor{KS_SauerwaldMeierKontinen2025} (\citeyear{KS_SauerwaldMeierKontinen2025}) (SMK), which also contains a broader introduction and motivation, and focuses on non-monotonic reasoning in the style of \citeauthor{KS_KrausLehmannMagidor1990} (\citeyear{KS_KrausLehmannMagidor1990}) (KLM). 

KLM (\citeyear{KS_KrausLehmannMagidor1990}) showed for classical propositional logic that cumulative entailment relations provide a stable theory of reasoning with multiple representations. Cumulative entailment relations are obtained by reasoning via the axiomatic system known as {\SystemC}. Furthermore, a cumulative entailment relation \( \nmableit \) can be represented by a cumulative model \( \pmC \), by which \( \formulaOne\nmableit\formulaTwo \) amounts (intuitively) to checking whether all minimal models of \( \formulaOne \) in \( \pmC \) are models of \( \formulaTwo \).
SMK (\citeyear{KS_SauerwaldMeierKontinen2025}) show that entailment via preferential models (which are specific cumulative models) satisfies {\SystemC} in the context of propositional dependence logic.
However, SMK provides no representation theorems for any kind of general cumulative reasoning in the context of team-based logics.

\MYParagraph{Contributions.}
This work 
establishes cumulative reasoning as a stable approach in the context of team-based logics.
 Specifically, we will encounter the cumulative counterparts of the following team-based logics:
\begin{itemize}
    \item Propositional logic with team-based semantics \hfill(\( \TPL \))
    \item Propositional dependence logic \hfill(\( \PDL \))
\end{itemize}
For each of these logics, we consider, both, reasoning via {\SystemC} and via cumulative models as in the approach by KLM.
This leads to {\SystemC}-based entailment relations for \( \PDL \), denoted by \( \PDLSystemC \). 
Also, we will consider \( \PDL \) entailment relations via cumulative models, denoted by \( \PDLCum \).
With \( \TPLCum \) and \( \TPLSystemC \), we denote the respective approaches for \( \TPL \).
We will show the following representation results:
\begin{description}
    \item[Representation Theorem for Cumulative {\PDL}:]
    \( \PDLSystemC \) and \( \PDLCum \) define the same entailment relations.
           
    \item[Representation Theorem for Cumulative {\TPL}:]
    \( \TPLSystemC \) and \( \TPLCum \) define the same entailment relations.
    Furthermore, we show that the novel class of asymmetric models are a sufficient subclass of cumulative models for representation.
\end{description}
The proofs for showing these results demonstrate, in a principled way, how similar relationships between these classes of entailment relations can be obtained for other logics without negation and material implication.

\section{Preliminaries}
\label{sec:background_team_based_logic}
In this paper, we consider propositional logics from a model-theoretic perspective.
We denote by $\Prop=\{\;p_i\mid i\in \mathbb{N}\;\}$ the countably infinite set of propositional variables. 
We consider propositional formulas in negation normal form, i.e., \PL-formulas are formed by the grammar, where $p\in\Prop$:
\[\formulaOne\Coloneqq p\mid \neg p \mid \bot\mid\top\mid \formulaOne\wedge\formulaOne\mid\formulaOne\vee\formulaOne.\]
 We write $\Prop(\formulaOne)$ for the set of variables occurring in $\formulaOne$.

\MYParagraph{Classical Propositional Logic ({\CPL}).}
For a non-empty finite subset $N\subseteq \Prop$ of propositional variables, one defines for valuations $v\colon N\to\{0,1\}$ over \( N \) and \PL-formulas \( \formulaOne \):
\[ \llbracket \formulaOne\rrbracket^c =
\{\, v\colon N\to\{0,1\}\mid v\models\formulaOne\}. \]
The valuation function $v$ is extended to the \emph{set of all \PL-formulas satisfying} $\Prop(\formulaOne)\subseteq N$, $\PL(N)$, in the usual way.
We denote by $\allModels{N}$ the set of all assignments over $N$. 
%
%
\begin{comment}
    Let $v:\Prop\to\{0,1\}$ be a valuation. For any classical $\formulaOne$, the truth relation $v\models\formulaOne$ is defined inductively as:
    \begin{itemize}
        \item $v\models p$ ~~iff~~ $v(p)=1$;
        \item $v\models \bot$ ~~is never the case;
        \item $v\models \top$ ~~is always the case;
        \item $v\models\neg\formulaOne$ ~~iff~~ $v\not\models\formulaOne$;
        \item $v\models\formulaOne\wedge\formulaTwo$ ~~iff~~ $v\models\formulaOne$ and $v\models\formulaTwo$;
        \item $v\models\formulaOne\vee\formulaTwo$ ~~iff~~ $v\models\formulaOne$ or $v\models\formulaTwo$.
    \end{itemize}
\end{comment}
%
%
%
We write $\formulaOne\models^c\formulaTwo$ for $ \llbracket \formulaOne\rrbracket^c \subseteq \llbracket \formulaTwo\rrbracket^c $ and $\formulaOne\equiv^c\formulaTwo$ if both $\formulaOne\models^c\formulaTwo$ and $\formulaTwo\models^c\formulaOne$ are true.
%
%
%

%

%
%
%

\begin{comment}

    \begin{theorem}\label{PL_usual_sem_expressive_complete}
        Let $N=\{p_1,\dots,p_n\}$. For every set $X\subseteq 2^N$, there is a \PL-formula $\formulaOne$
        such that $\llbracket \formulaOne\rrbracket_N^c=X$.
    \end{theorem}
    \begin{proof}
        Let 
        \[\formulaOne=\bigvee_{v\in X}(p_1^{v(1)}\wedge\dots\wedge p_n^{v(n)}),\]
        where $v(i)$ is short for $v(p_i)$, $p_i^{1}=p_i$ and $p_i^0=\neg p_i$ for any $1\leq i\leq n$, and in particular $\bigvee\emptyset:=\bot$. For any valuation $u:N\to\{0,1\}$, we have that 
        \[u\models \formulaOne\iff u\models p_1^{v(1)}\wedge\dots\wedge p_n^{v(n)}\text{ for some }v\in X\iff u=v\in X.\] 
    \end{proof}

    \begin{corollary}\label{PL_usual_sem_dnf}
        Let $N=\{p_1,\dots,p_n\}$. For every classical formula $\formulaOne(p_1,\dots,p_n)$, we have that
        \[\formulaOne\equiv^c\bigvee_{v\in\llbracket \formulaOne\rrbracket^c_N}(p_1^{v(1)}\wedge\dots\wedge p_n^{v(n)}).\]
    \end{corollary}
\end{comment}

\MYParagraph{Prop.\ Logic with Team Semantics ({\TPL}).}
Next, we define \emph{team semantics} for \PL-formulas (cf. \cite{HannulaKVV15,YangV16}). A \emph{team $X$} is a set of valuations for some finite $N\subseteq \Prop$. The \emph{domain $N$ of $X$} is denoted by  $\dom(X)$, and $\mathbb{T}_{N}$ denotes the set of all such teams.%

\begin{definition}[Team semantics of \PL]
    Let $X$ be a team. For any \PL-formula $\formulaOne$ with $\dom(X)\supseteq \Prop(\formulaOne)$, the \emph{satisfaction relation}, $X\models\formulaOne$, is defined inductively as:
    \begin{align*}
        &X\models p&&\text{if for all }v\in X: v\models p,\\
        &X\models\neg p&&\text{if for all } v\in X: v\not\models p,\\
        &X\models \top&&\text{is always the case},\\
        &X\models \bot&&\text{if } X=\emptyset,\\
        &X\models\formulaOne\wedge\formulaTwo&&\text{if } X\models\formulaOne \text{ and } X\models\formulaTwo,\\
        &X\models\formulaOne\vee\formulaTwo&&\text{if there exist } Y,Z\subseteq X\\
        &&&\text{s.t. }X=Y\cup Z, Y\models\formulaOne, \text{ and }Z\models\formulaTwo.
    \end{align*}
\end{definition}
The set of all teams \( X \) with \( X\models\formulaOne \) is denoted by \( \modelsOf{\formulaOne}^t \).
For any two \PL-formulas \( \formulaOne,\formulaTwo \), we write $\formulaOne\models^t\formulaTwo $ if \( \modelsOf{\formulaOne}^t \subseteq \modelsOf{\formulaTwo}^t \).
Write $\formulaOne\equiv^t\formulaTwo$ if both $\formulaOne\models^t\formulaTwo$ and $\formulaTwo\models^t\formulaOne$ are true.
We define the following properties for a formula~$\formulaOne$:
\begin{itemize}
    \item $X\models \formulaOne \iff \text{for all } v\in X,~\{v\}\models\formulaOne$. \hfill {\small(\textbf{Flatness})}
    \item $\emptyset \models \formulaOne$.\hfill{\small(\textbf{Empty team})}
    \item If $X \,{\models}\, \formulaOne$ and $Y\,{\subseteq}\, X$, then ${Y\models \formulaOne}$.\hfill{\small(\textbf{Downward closure})}
\end{itemize}
\begin{proposition}\label{prop:tpl_pincl_properties}
    {\TPL} has the properties flatness, empty team, and downward closure. 
\end{proposition}

Due to the flatness property, logical entailment of propositional logic with team-based semantics $\models^t$ and %
classical semantics $\models^c$ coincide. 
\MYParagraph{Propositional Dependence Logic ({\PDL}).}
\label{sec:pdl_pincl}
A \emph{(propositional) dependence atom} is a string %
$\dep{\vec{a},b}$, 
in which $\vec{a}=a_1,\dots,a_k$ and $b$ are propositional variables from \Prop. 
A team $X$ \emph{satisfies a dependence atom}, $X \models \dep{\vec{a},b}$, if for all $v, v' \in X$, $v(\vec{a})=v'(\vec{a})$ implies $v(b)=v'(b)$.
A dependence atom where the first component is empty will be abbreviated as $\dep{p}$ and called a \emph{constancy atom}. %
The language of \emph{propositional dependence logic} (\pdl) is defined as \PL-formulas extended by dependence atoms.
\begin{example}\label{example_dep_atm_propositional}
    Consider the team $X$ over $\{p,q,r\}$ defined by:\smallskip

        \noindent\;\;\begin{tabular}{cccc}
            \toprule
            &$p$&$q$&$r$\\\midrule
            $v_1$&$1$&$0$&$0$\\
            $v_2$&$0$&$1$&$0$\\
            $v_3$&$0$&$1$&$0$\\
            \bottomrule
        \end{tabular}\hfill
    \begin{minipage}[c]{0.3\textwidth}
        \vspace{0pt}
    Here, we have that $X\models\dep{p,q}$ and $X\models\dep{r}$. 
    Moreover, we have that $X\models\dep{p}\vee \dep{p}$, however it is true that $X\not \models\dep{p}$ as the value of $p$ is not overall constant.     
    \end{minipage}
\end{example}
\begin{comment}
    As illustrated in the above example, in general, propositional dependence atoms with multiple arguments can be easily defined using constancy atoms.

    \begin{fact}
        $\displaystyle\dep{p_1\dots p_n,q}\equiv\bigvee_{v\in 2^N}(p_1^{v(1)}\wedge\dots\wedge p_n^{v(n)}\wedge\dep{q})$, where $N=\{p_1,\dots,p_n\}$.
    \end{fact}
    
\end{comment}

%
%
\begin{proposition}\label{prop:pdl_pincl_properties}
    \PDL has the empty team and the downward closure property, but not the flatness property. 
\end{proposition}

\MYParagraph{Generic View on Logics.} 
Some parts of this paper require a generic perspective on logics, which we discuss next.
A \emph{satisfaction system} is a triple \( \ssSystem = \tuple{\ssFormulas,\ssInt,\models} \), where \( \ssFormulas \) is the set of \emph{formulas}, \( \ssInt \) is the set of \emph{interpretations}, and  \( {\models} \subseteq \ssInt \times \ssFormulas \) is the \emph{satisfaction relation}. 
We write $\modelsOf{\formulaOne}^\ssSystem=\{ \omega \in \ssInt \mid \omega \models \formulaOne \} $ for the set of all models of the formula $\alpha\in\ssFormulas$.
An \emph{entailment relation} for a satisfaction system is a relation \( {\Vdash} \subseteq \ssFormulas \times \ssFormulas\) such that %
$\alpha \Vdash \gamma$ if and only if $\beta \Vdash \gamma$, 
whenever $\modelsOf{\formulaOne}^\ssSystem=\modelsOf{\formulaTwo}^\ssSystem $.
A satisfaction system \( \ssSystem \) together with an entailment relation \( \Vdash \) is called a \emph{logic} and  denoted by
\( \ssLogic = \tuple{\ssFormulas,\ssInt,\models,\Vdash} \).
The propositional logics discussed in this section fit into this general model-theoretic view for each \( N \subseteq \Prop \) as follows:
\begin{center}
    \begin{tabular}{@{}l@{\hskip0.05em}l@{}}
        \( \CPL_{N} \)  
        & \( {=}\, \tuple{\PL(N),\allModels{N},\models,\models^c} \)\\
        \( \TPL_{N} \) 
        & \( {=}\, \tuple{\PL(N),\allTeams{N},\models,\models^t} \)\\
        \( \PDL_{N} \) 
        & \( {=}\, \tuple{\pdlN, \allTeams{\!N},\models,\models^t}\)
    \end{tabular}
\end{center}
When there is no ambiguity, we will write $\models$ instead of $\models^t$.
Moreover, we will use \( \CPL \) to denote the class consisting of all logics \( \CPL_{N} \) for any non-empty finite subset $N\subseteq \Prop$.
The classes \( \TPL \) and \( \PDL \) are defined analogous.

\section{{\SystemC} and Cumulative Models}
\label{sec:background_klmstyle}
We consider the construction of  entailment relations.

\MYSubsection{{\SystemC}}
 We make use of the following rules for calculi:
\begin{center}
	\vspace{-1em}
	\begin{minipage}[b]{0.51\linewidth}
		\begin{align}
			&\frac{\formulaOne\equiv\formulaTwo\hspace{0.5cm}\formulaOne\nmableit\formulaThree}{\formulaTwo\nmableit\formulaThree} \tag{LLE}\label{pstl:LLE}\\[0.25em]
			&\frac{\formulaOne\land\formulaTwo\nmableit\formulaThree\hspace{0.5cm}\formulaOne\nmableit\formulaTwo}{\formulaOne\nmableit\formulaThree} \tag{Cut}\label{pstl:Cut}
		\end{align}
	\end{minipage}\begin{minipage}[b]{0.45\linewidth}
		\begin{align}
			&\frac{\formulaOne\models\formulaTwo\hspace{0.5cm}\formulaThree\nmableit\formulaOne}{\formulaThree\nmableit\formulaTwo} \tag{RW}\label{pstl:RW}\\[0.25em]
			&\frac{\formulaOne\nmableit\formulaTwo\hspace{0.5cm}\formulaOne\nmableit\formulaThree}{\formulaOne\land\formulaTwo\nmableit\formulaThree} \tag{CM}\label{pstl:CM}
		\end{align}
	\end{minipage}
\end{center}
Note that \( \models \) is a placeholder for the entailment relation \( \Vdash \) of an underlying logic \( \ssLogic = \tuple{\ssFormulas,\ssInt,\models,\Vdash}  \), and \( \equiv \) is the respective semantic equivalence from $\ssLogic$.
{\SystemC}  consists of the rules  \eqref{pstl:RW}, \eqref{pstl:LLE},  \eqref{pstl:CM}, \eqref{pstl:Cut} and reflexivity  (\textlabel{Ref}{pstl:Ref}) $\formulaOne \nmableit \formulaOne$ for all formulas \cite{KS_KrausLehmannMagidor1990,KS_Gabbay1984}.
We say that an entailment relation \( \nmableit \) \emph{satisfies} {\SystemC} if \( \nmableit \) is closed under all rules of {\SystemC}.

\MYSubsection{Relational Models}
For a relation \( {\pmR} \subseteq \mathcal{S} \times \mathcal{S} \) on a set \( \mathcal{S} \) and a subset \( S \subseteq \mathcal{S} \), an element \( s \in S \) is called \emph{minimal in \( S \) with respect to \( \pmR \)} if for each \( s' \in S \) holds  \( \lnot(s' \pmR s) \). 
Then, \( \minOf{S}{{\pmR}} \)  is the set of all \( s \in S \)  that are minimal in \( S \) with respect to \( \pmR \).
Moreover, for a set of interpretations \( M \) and a formula \( \varphi \) of  \( \ssLogic = \tuple{\ssFormulas,\ssInt,\models,\Vdash}  \), we write \( M \models \varphi \) if for all \( \omega\in M \) we have that \( \omega \models \varphi \).
\begin{definition}[Shoman~\citeyear{KS_Shoham1988}, Dix and Makinson~\citeyear{KS_Makinson1989,KS_DixMakinson1992}]
	Let \( \ssLogic = \tuple{\ssFormulas,\ssInt,\models,\Vdash}  \) be a logic.
	A \emph{relational model} for \(  \ssLogic  \) is a triple \( \pmM=\tuple{\pmS,\pmL,\pmR} \) where \( \pmS \) is a set, \( \pmL \colon \pmS \to \mathcal{P}(\Omega) \), and \( \pmR \) is a binary relation on \( \pmS \).
\end{definition}
We say \( S \subseteq \pmS \) is \emph{smooth} if for each \( s \in S \), we either have that \( s\in \minOf{S}{\pmR} \), or there exists a state \( s'\in \minOf{S}{\pmR} \) with \( s' \pmR s \).
For \( \formulaOne \in \ssFormulas \) and \( s\in\pmS \), we denote by \( \statesOf{\formulaOne}=\{\, s \in \pmS \mid \ell(s) \models \formulaOne\, \} \) the set of states that satisfy \( \formulaOne \). %
With \( \minOf{\modelsOf{\formulaOne}}{\pmR} = \bigcup\{\, \pmL(s) \mid s\in \minOf{\statesOf{\formulaOne}}{\pmR}  \,\} \), we denote the set of all interpretations that appear in \( \pmL(s) \) for any minimal state \( s \) that satisfy \( \formulaOne \).
We will deal, in this paper, with specific types of relational models, which we define in the following.

\begin{definition}
	A relational model \( \pmC=\tuple{\pmS,\pmL,\pmR} \) for a logic \( \ssLogic = \tuple{\ssFormulas,\ssInt,\models,\Vdash}  \) is called 
    \emph{cumulative} if for all \( \formulaOne \in \ssFormulas \) the set \( \statesOf{\formulaOne} \) is smooth.
    A cumulative model \( \pmC \) is \emph{strong} if \( \pmR \) is asymmetric ($x R y$ implies $\lnot(y R x)$) and \( \minOf{\statesOf{\formulaOne}}{\pmR} \) has exactly one element for each formula \( \formulaOne \in \ssFormulas \).
\end{definition}
Because cumulative models satisfy smoothness, they avoid the problem of reasoning via arbitrary relational models, in which the set $\minOf{\modelsOf{\formulaOne}}{\pmR}$ might be empty, even when $\statesOf{\formulaOne}$ is non-empty. 
Entailment relations, which are induced by a relational model, are defined as follows.
\begin{definition}
	Let \( \ssLogic = \tuple{\ssFormulas,\ssInt,\models,\Vdash}  \) be a logic.
	The \emph{entailment relation} \( {\nmableit_{\pmM}} \subseteq \ssFormulas \times \ssFormulas \) for a relational model \( \pmM=\tuple{\pmS,\pmL,\pmR} \) for \(  \ssLogic \) is given by
	\begin{equation*}
		\formulaOne \nmableit_{\pmM} \formulaTwo \text{ if }  \minOf{\modelsOf{\formulaOne}}{\pmR} \subseteq \modelsOf{\formulaTwo}\ .
	\end{equation*}
\end{definition}
An entailment relation \( {\nmableit} \subseteq \ssFormulas \times \ssFormulas \) is called \emph{(strongly) cumulative} if there is a (strong) cumulative model \( \pmC \) for \( \ssLogic \) such that  \( {\nmableit} =   {\nmableitC}  \). 

\MYSubsection{Classes of Entailment Relations}
If \( \logicFont{L} \) is a class of logics (such as \PDL, \TPL, or \CPL),  we denote with \( \pmLogicSystemC{\logicFont{L}} \)  the class of all entailment relations \( {\nmableit} \) for \( \tuple{\ssFormulas,\ssInt,\models,\Vdash} \in \logicFont{L} \) that satisfy {\SystemC}. 
Similarly, we use  \( \pmLogicCum{\logicFont{L}} \) (\( \pmLogicStrongCum{\logicFont{L}} \)) for the class of all (strong) cumulative entailment relations.

\section{Representation Theorems}
In this section, we will establish representation theorems for cumulative logics, i.e., logics that satisfy {\SystemC}. A classic result is that cumulative classical propositional logics coincide with classical propositional logics that satisfy {\SystemC}.
\begin{proposition}[KLM, \citeyear{KS_KrausLehmannMagidor1990}]
	\label{thm:CPLCPresepresention}
	$ \CPLCum = \CPLStrongCum = \pmLogicX{\CPL}{\logicFont{c}}$.
\end{proposition}

Here we will prove that when basing {\SystemC} on propositional dependence logic, there is also the same connection between cumulative logics and cumulative models.
We will also show that logics based on preferential and cumulative models for $\TPL$ correspond to {\SystemC} logics for $\CPL$.

\subsection{Propositional Dependence Logic}
\label{sec:SystemCPDL}
We will show the following representation result that establishes an equivalence between the class of cumulative propositional dependence logics and propositional dependence logics that satisfy {\SystemC}.
\begin{theorem}
	\label{thm:PDLCPresepresention}
	$ \PDLCum = \PDLStrongCum = \PDLSystemC $.
\end{theorem}
First, we like to remark that the original proof of Proposition~\ref{thm:CPLCPresepresention} makes heavy use of classical negation and material implication of the underlying propositional logic.
As neither classical negation nor material implication are available (and also not definable) in \( \PDL \), the proof of Theorem~\ref{thm:PDLCPresepresention} requires a new approach that circumvents the usage of implication and negation. We present the proof in the following, and mark those steps that are borrowed from KLM (\citeyear{KS_KrausLehmannMagidor1990}).

Our first observation is that entailment relations in $\PDLCum$ satisfy {\SystemC}. This is a novel result as the proof of {\SystemC} satisfaction for $\CPLPref$ given by KLM does not carry over to preferential propositional dependence logic.
\begin{proposition}
    \label{prop:systemCteams}
    \( \PDLCum \subseteq  \PDLSystemC \).
\end{proposition}
\begin{proof}[Proof (idea)]
	One checks case by case that every postulate of {\SystemC} is satisfied.
    A full proof is provided in the accompanied supplementary material.
\end{proof}

In the remainder of this section, we will show that every entailment relation \( {\nmableit} \in \PDLSystemC \) is strongly cumulative. %
For that, we will first show that one can characterize all the consequences of a formula \( \formulaOne \) via \( \nmableit \) semantically.

For a given formula \( \formulaOne \), we denote the set of all teams that satisfy all consequences of \( \formulaOne \) with respect to \( \nmableit \), by
\begin{align*}
    \mathrm{Norm}(\formulaOne,{\nmableit}) =  \{\, X \mid X \models \formulaTwo \text{ for all } \formulaTwo \text{ with } \formulaOne \nmableit \formulaTwo \,\} \ .
\end{align*}
In the following, we will show that for every \( \formulaOne \) the set \( \mathrm{Norm}(\formulaOne,{\nmableit}) \)  characterize semantically the consequences of \( {\nmableit} \) drawn from \( \formulaOne \).
For all formulas \( \formulaOne \), all sets of formulas \( F \subseteq \ssFormulas \), and all sets of teams \( M \)  we define:
\begin{align*}
     \Cableit(\formulaOne) & = \{\, \formulaTwo \mid \formulaOne \nmableit \formulaTwo \,\}, \quad \Cn(\formulaOne) = \{ \,\formulaTwo \mid \formulaOne \models \formulaTwo \,\}, \\
     \Th(M) & = \{\, \formulaOne \in \ssFormulas \mid X \models \formulaOne  \text{ for all } X \in M \,\},\\
     \Cn(F) & = \{\, \formulaTwo \mid \formulaOne \models \formulaTwo \text{ for all } \formulaOne \in F \,\}.
\end{align*}
The following lemma provides basic insights into the closures \( \Cableit \) and \( \Cn\), and the theory operator \( \Th \).
Due to space constraints, the proof is given in the supplementary material.%
\begin{lemma} 
    \label{lem:closures}
    The following statements are true:
	\begin{enumerate}[(a)]
		\item 
		For \( {\nmableit} \in \PDLSystemC \) we have that \( \Cableit(\formulaOne) = \Cn(\Cableit(\formulaOne)) \).
		\item 
		For each formula \( \formulaOne \) there exists an \( M \subseteq \mathcal{P}(\mathbb{T}_{N}) \) such that \( \Cn(\formulaOne) = \Th(M) \) and \( M = \modelsOf{\formulaOne} \).
		\item 
		For each set of formulas \( F \) with \( F=\Cn(F) \) there exists a formula \( \formulaOne_{F} \) such that \( F = \Cn(\formulaOne_{F}) \) is true.
	\end{enumerate}
\end{lemma}

 Note that the proof of Lemma~\ref{lem:closures}  relies on the  fact that $\pdl$ is expressively complete for team properties that are downward closed and contain the empty team \cite{Yang_dissertation}. Now, we show the central insight that \( \mathrm{Norm}(\formulaOne,{\nmableit}) \) characterizes all consequences of \( \formulaOne \) by \( {\nmableit} \).
\begin{theorem}[Definability of Consequences]
    \label{thm:definablity}
For all entailment relations \( {\nmableit} \in \PDLSystemC \)  and all formulas \( \formulaOne \) we have that:
\begin{enumerate}[(a)]
	\item For all formulas \( \formulaTwo \) we have that:\\
		\(\formulaOne \nmableit \formulaTwo\) 
		if and only if \(\mathrm{Norm}(\formulaOne,{\nmableit}) \models \formulaTwo\).
    \item There is a \( \theta_{\formulaOne} \) such that for all formulas \( \formulaTwo \) we have that:\\
    \(\formulaOne \nmableit \formulaTwo\) if and only if \(\theta_{\formulaOne} \models \formulaTwo\).
\end{enumerate}
\end{theorem}
\begin{proof}
    Let \( {\nmableit} \) and \( \formulaOne \) be as above.
    We let \( \theta_{\formulaOne} \) be a formula such that \( \modelsOf{\theta_{\formulaOne}} = \mathrm{Norm}(\formulaOne,{\nmableit})  \) which exists due to Lemma~\ref{lem:closures}.

\smallskip
\noindent[Case \enquote*{\( \formulaOne \nmableit \formulaTwo \Rightarrow  \mathrm{Norm}(\formulaOne,{\nmableit}) \models \formulaTwo    \)}] If \( \formulaOne \nmableit \formulaTwo \), then we obtain \( \theta_{\formulaOne} \models \formulaTwo   \) by definition of \( \mathrm{Norm}(\formulaOne,{\nmableit}) \).
\smallskip

\noindent[Case \enquote*{\( \theta_{\formulaOne} \models \formulaTwo  \Rightarrow \formulaOne \nmableit \formulaTwo \)}]
By Lemma~\ref{lem:closures}, there exist \( M \subseteq \mathcal{P}(\mathbb{T}_{N}) \) such that \( \mathop{C}_{{\nmableit}}(\formulaOne) = \Th(M) \) is true.

We show that \( M \subseteq \mathrm{Norm}(\formulaOne,{\nmableit}) \) is true.
Suppose not, i.e., there is an \( X \in M \) such that  \( X \notin \mathrm{Norm}(\formulaOne,{\nmableit}) \). From the latter, we obtain that there is a formula \( \formulaThree \) with \( X \not\models \formulaThree \) and \( \formulaOne \nmableit \formulaThree \).
Now, from \( \mathop{C}_{{\nmableit}}(\formulaOne) = \Th(M) \), we can deduce \( X \models \delta \) for all \( \delta \in \mathop{C}_{{\nmableit}}(\formulaOne) \). Consequently, \( X \models \formulaThree \), as we have that \( \formulaThree \in \mathop{C}_{{\nmableit}}(\formulaOne) \). This shows \( M \subseteq \mathrm{Norm}(\formulaOne,{\nmableit}) \).

We then obtain that \( \Th(\mathrm{Norm}(\formulaOne,{\nmableit})) \subseteq \Th(M)  \) by employing \( M \subseteq \mathrm{Norm}(\formulaOne,{\nmableit}) \).
Consequently, from \( \mathrm{Norm}(\formulaOne,{\nmableit}) \models \formulaTwo \),  we deduce that \( M \models \formulaTwo \).
Thus, by employing \( \modelsOf{\theta_{\formulaOne}} = \mathrm{Norm}(\formulaOne,{\nmableit}) \) and \( \mathop{C}_{{\nmableit}}(\formulaOne) = \Th(M) \), from which we have that \( \theta_{\formulaOne} \models \formulaTwo \) implies \( \formulaOne \nmableit \formulaTwo \). \qedhere
\end{proof}

Next, we construct a cumulative model that captures \( \nmableit \). The construction is inspired by the construction of \ksCiteAY{KS_KrausLehmannMagidor1990}. 
Let \( {\sim} \) be the equivalence relation, respectively \( \preceq_{\nmableit} \) be the relation, defined by 
\begin{align*}
    \formulaOne \sim \formulaTwo \quad& \text{ if } \formulaOne \nmableit \formulaTwo \text{ and } \formulaTwo \nmableit \formulaOne,\\
  [\formulaOne]_{{\sim}} \preceq_{\nmableit} [\formulaTwo]_{{\sim}} \quad& \text{ if there exists }  \formulaThree \in [\formulaOne]_{{\sim}} \text{ with } \formulaTwo \nmableit \formulaThree \ .
\end{align*}
Define \( \pmC_{\nmableit}=\tuple{\pmS,\pmL,{\pmR}} \), with $\pmS \coloneqq \pdl / {\sim}$ as follows:
\begin{align*}
  \pmL &\coloneqq \{\, ([\formulaOne]_{{\sim}}, \mathrm{Norm}(\formulaOne,{\nmableit})) \mid \formulaOne \in \pdl  \,\} \ , \\
  {\pmR} &\coloneqq \{ \, ([\formulaOne]_{{\sim}},[\formulaTwo]_{{\sim}}) \mid  [\formulaOne]_{{\sim}} \neq [\formulaTwo]_{{\sim}} \text{ and } [\formulaOne]_{{\sim}} \preceq_{\nmableit} [\formulaTwo]_{{\sim}} \,\} \ .
\end{align*}
The states of \( \pmC_{\nmableit} \) are the equivalence classes of \( {\sim} \). The function \( \pmL \) assign to every state \( [\formulaOne]_{{\sim}} \) the normal worlds of \( \formulaOne \), and we have that \( [\formulaOne]_{{\sim}} \pmR [\formulaTwo]_{{\sim}} \) if \( [\formulaOne]_{{\sim}} \) and \(  [\formulaOne]_{{\sim}} \) are different and \(  [\formulaOne]_{{\sim}} \preceq_{\nmableit} [\formulaTwo]_{{\sim}} \) is true. 

\begin{proposition}
    \( \PDLSystemC \subseteq \PDLStrongCum. \)
    \label{prop:PDLRepSystemCconstruction}
\end{proposition}
\begin{proof}
First, we show that \( \pmC_{\nmableit} \) is a strong cumulative model.
One sees easily that \( \pmC_{\nmableit} \) is a relational model and \( {\pmR} \) is asymmetric.
It suffices to show that for every formula \( \formulaOne \) the state \( [\formulaOne]_{{\sim}} \) is the unique element of \( \minOf{\statesOf{\formulaOne}}{\pmR} \).
Suppose not, i.e., there is a state \( [\formulaTwo]_{{\sim}} \in \minOf{\statesOf{\formulaOne}}{\pmR}  \)  different from \( [\formulaOne]_{{\sim}} \).
By definition, we have \( \pmL([\formulaTwo]_{{\sim}}) = \mathrm{Norm}(\formulaTwo,{\nmableit}) \) and hence, that  \( \mathrm{Norm}(\formulaTwo,{\nmableit}) \subseteq \modelsOf{\formulaOne}^t \) is true. 
We obtain from the latter that \( \formulaTwo \nmableit \formulaOne \) is true by employing Theorem~\ref{thm:definablity}.
From \( \formulaOne \in [\formulaOne]_{{\sim}} \) and \( \formulaTwo \nmableit \formulaOne \), we obtain \( [\formulaOne]_{{\sim}} \preceq_{\nmableit} [\formulaTwo]_{{\sim}} \).
Consequently, we have that \( [\formulaOne]_{{\sim}} \pmR [\formulaTwo]_{{\sim}} \), because of \( [\formulaOne]_{{\sim}} \preceq_{\nmableit} [\formulaTwo]_{{\sim}} \) and \( [\formulaOne]_{{\sim}} \neq [\formulaTwo]_{{\sim}} \).
However, \( [\formulaOne]_{{\sim}} \pmR [\formulaTwo]_{{\sim}} \) is a contradiction to \( [\formulaTwo]_{{\sim}} \in \minOf{\statesOf{\formulaOne}}{\pmR}  \).

To complete the proof, we show that \( \nmableit \) and \( \nmableit_{\pmC_{\nmableit}} \) coincide.
Due to Theorem~\ref{thm:definablity}, we have that \( \formulaOne \nmableit \formulaTwo \)
if and only if \( \mathrm{Norm}(\formulaOne,{\nmableit}) \models \formulaTwo \).
By construction of \( \pmC_{\nmableit} \), it is true that \( \minOf{\modelsOf{\formulaOne}^t}{R} = \mathrm{Norm}(\formulaOne,{\nmableit}) \).
By definition, we get \( \formulaOne \nmableit_{\pmC_{\nmableit}} \formulaTwo \) exactly when  \( \mathrm{Norm}(\formulaOne,{\nmableit}) \models \formulaTwo \).
Consequently, we have that \( \formulaOne \nmableit \formulaTwo \) if and only if \( \formulaOne \nmableit_{\pmC_{\nmableit}} \formulaTwo \).
\end{proof}

Because  \( \PDLStrongCum \subseteq \PDLCum \) holds, we obtain Theorem~\ref{thm:PDLCPresepresention} from Proposition~\ref{prop:PDLRepSystemCconstruction} and Proposition~\ref{prop:systemCteams}.

\subsection{\!\!\!Propositional Logic with Team-based Semantics}
We show that classes of {\SystemC}-based entailment relations and cumulative entailment relations coincide with cumulative entailment relations for propositional logic with classical semantics and the following kind of cumulative models. 

\begin{definition}
	A relational model \( \pmC=\tuple{\pmS,\pmL,\pmR} \) for a logic \( \ssLogic \) is called 
	\emph{asymmetric} if \( \pmC \) is cumulative, $\pmP$ is asymmetric and for each state $s\in S$, the set $\pmL(s)$ is a singleton.
\end{definition}
With \( \TPLAsym \) we denote the class of entailment relations based on asymmetric models over \( \TPL \).

\begin{theorem} 
        \(\TPLAsym = \CPL^\logicFont{x} = \CPL^\logicFont{y}  = \TPL^\logicFont{x} = \TPL^\logicFont{y}\)
    is true for all $\logicFont{x},\logicFont{y}\in\{\, \logicFont{cuml},\  \logicFont{cuml[str]},\ \logicFont{c} \,\}$.
\end{theorem}

\begin{proof}
Note first that Proposition~\ref{thm:CPLCPresepresention} implies that \(  \CPLSystemC =  \CPLCum=\CPLStrongCum \). Furthermore, we have that \( \TPLStrongCum \subseteq \TPLCum \) and \(\, \TPLSystemC = \CPLSystemC \,\) are immediate consequences of the definitions. %
Finally, the analogue of Prop.~\ref{prop:systemCteams}   can be used to show that \( \TPLCum \,{\subseteq}\, \TPLSystemC \) is true.

Let us then show that \(\, \CPLStrongCum \,{\subseteq}\, \TPLStrongCum \,\).
Let \( \nmableitC \in \CPLStrongCum \) and \( \pmC = \tuple{\pmS,\pmL,\pmR} \) be the respective strong cumulative model for \( \CPL \). We reinterpret \(\pmC\) as a strong cumulative model \( \pmC' = \tuple{\pmS',\pmL',\pmR'} \) for \( \TPL \) with \( \pmS' = \pmS, \), \( \pmL' = \{\, (s,\{ \pmL(s)\}) \mid s \in \pmS \,\}\), \(  \pmR'  = {\pmR}. \)
Because in \( \TPL \) is holds that ($\star$) \(  \{\nu\} \models \formulaOne  \) if and only if \( \nu \models^{c} \formulaOne \) for all singleton teams \( \{\nu\} \), we obtain then \( {\nmableitC} = {\nmableit_{\pmC'}} \). 
Next, we show that $\CPLStrongCum \subseteq \TPLAsym $ is true. 
Note that any strong cumulative model \( \pmC = \tuple{\pmS,\pmL,\pmR} \) for \( \CPL \) can be reinterpreted ($\pmL(s)=\{\,\nu_1,\ldots,\nu_m\,\}$ becomes $\pmL(s)=\{\,\{\nu_1,\ldots,\nu_m\}\,\}$) as an asymmetric model for \( \TPL \), hence by the flatness property of \( \TPL \)-formulas and ($\star$) above the claim follows.
Finally, we show that \( \TPLAsym \subseteq \CPLCum \) holds.
Note that any asymmetric model \( \pmC = \tuple{\pmS,\pmL,\pmR} \) for \( \TPL \) 
can be interpreted ($\pmL(s)=\{\,\{\nu_1,\ldots,\nu_m\}\,\}$ becomes  $\pmL(s)=\{\,\nu_1,\ldots,\nu_m\,\}$) as a cumulative model $\pmC$ for \( \CPL \). By the flatness property of \( \TPL \)-formulas and ($\star$) above it follows that  \( {\nmableitC} = {\nmableit_{\pmC'}} \). 
\qedhere
\end{proof}

\section{Conclusion}
\label{sec:conclusion}

In this paper, we proved representation theorems between entailment relations based on {\SystemC} and based on cumulative reasoning for propositional dependence logic and propositional logic with team semantics.
Hence, they are forming, in the context of these logics, a stable class which is justified to be denoted as \emph{cumulative logics}.
The results from SMK~({\protect\citeyear{KS_SauerwaldMeierKontinen2025}}) and KLM~({\protect\citeyear{KS_KrausLehmannMagidor1990}}), and the novel results from this paper, yield a diverse landscape as shown in  Figure~\ref{fig:landscape}. 
\begin{figure}[t]
\small
    \begin{tikzpicture}
    \def\mydist{0.8cm}
    \def\mydistP{2*\mydist}
    \def\myheight{0.5cm}
        \begin{scope}
        \node[anchor=west,minimum height=\myheight,draw,rectangle,rounded corners] (TPLPrefStar) at (0,0) {\( \pmLogicX{\TPL}{\logicFont{pref}\ensuremath{[\star]}} \)};
        \node[anchor=west,minimum height=\myheight,draw,rectangle,rounded corners] (TPLPrefCapP) at ([yshift=\mydist]TPLPrefStar.west) {\( \TPLPref \cap \pmLogicX{\TPL}{\logicFont{p}} \)};
        
        \node[anchor=west,minimum height=\myheight,draw,rectangle,rounded corners] (TPLp) at ([yshift=\mydist]TPLPrefCapP.west) {\( \TPLSystemP\)\,{=}\,\( \CPLSystemP \)\,{=}\,\( \CPLPref \)};
        \node[anchor=west,minimum height=\myheight,draw,rectangle,rounded corners] (TPLPref) at ([yshift=\mydist]TPLp.west) {\( \TPLPref \)};
        \node[anchor=west,minimum height=\myheight,draw,rectangle,rounded corners] (TPLc) at ([yshift=\mydist]TPLPref.west) {\( \TPLAsym \)\,{=}\,\( \TPLSystemC \)\,{=}\,\( \CPLSystemC \)\,{=}\,\( \TPLCum \)\,{=}\,\( \CPLCum \)};

        \draw (TPLPrefStar) edge [] (TPLPrefCapP);
        \draw (TPLPrefCapP) edge [dotted,thick] (TPLp);
        \draw (TPLp) edge [dotted,thick] (TPLPref);
        \draw (TPLp) edge [bend right=12] (TPLc);
        \draw (TPLPref) edge [dotted,thick] (TPLc);
        \end{scope}
        \begin{scope}
            \node[anchor=east,minimum height=\myheight,draw,rectangle,rounded corners] (PDLPrefStar) at (\columnwidth,0) {\( \pmLogicX{\PDL}{\logicFont{pref}\ensuremath{[\star]}} \)\,{=}\,\( \pmLogicX{\PDL}{\logicFont{pref}\ensuremath{[\triangle]}} \)\,{=}\,\( \PDLPref \cap \PDLSystemP \)};
        \node[anchor=east,minimum height=\myheight,draw,rectangle,rounded corners,minimum width=1.2cm] (PDLp) at ([xshift=-6em,yshift=\mydistP]PDLPrefStar.east) {\( \PDLSystemP \)};
        
        \node[anchor=east,minimum height=\myheight,draw,rectangle,rounded corners,minimum width=1.2cm] (PDLPref) at ([yshift=\mydistP]PDLPrefStar.east) {\( \PDLPref \)};
        \node[anchor=east,minimum height=\myheight,draw,rectangle,rounded corners] (PDLc) at ([yshift=\mydistP]PDLPref.east) {\( \PDLSystemC \)\,{=}\,\( \PDLCum \)};

        \draw (PDLPrefStar) edge [dotted,thick] (PDLp);
        \draw (PDLPrefStar) edge [] (PDLPref);
        \draw (PDLp) edge [] (PDLc);
        \draw (PDLPref) edge [] (PDLc);
        \end{scope}
    \end{tikzpicture}
    \vspace{-0.5cm}
    \caption{Landscape of classes of entailment relations. 
    For not defined classes, consult  SMK~({\protect\citeyear{KS_SauerwaldMeierKontinen2025}}) or the supplemental material.
    \vspace{-0.5cm}}
    \label{fig:landscape}
\end{figure}
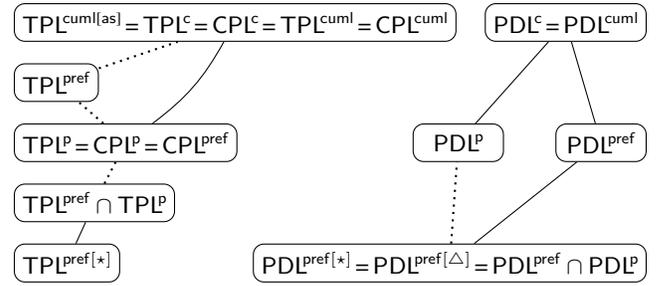

On our future agenda, we will extend our work to other team-based logics. Furthermore, we will study subclasses of the rich landscape of cumulative reasoning, such as preferential reasoning (KLM \citeyear{KS_KrausLehmannMagidor1990}), c-inference~\cite{KS_Kern-Isberner2001}, lexicographic inference~\cite{KS_Lehmann1995a} or System~W~\cite{KS_KomoBeierle2022}.

\clearpage
\bibliographystyle{kr}
\bibliography{merged.bib}

\clearpage
\appendix
\twocolumn[\begin{center}
    {\LARGE Supplemental Material}
\end{center}\vspace{0.5em}]
\setcounter{theorem}{0}
\counterwithin{theorem}{section}

\noindent This part contains supplemental information on the relationship among classes of entailment relations and proofs.

\section{\!\!\!Landscape: Classes of Entailment Relations}

We summarise hierarchies of classes of entailment relations provided by SMK~({\protect\citeyear{KS_SauerwaldMeierKontinen2025}}), KLM~({\protect\citeyear{KS_KrausLehmannMagidor1990}}) and this paper. Boxes contain known equivalent classes, lines stand for inclusions, whereby solid lines are strict inclusions.
We start with the hierarchy of classes of entailment relations over the language~\( \PL \):\\
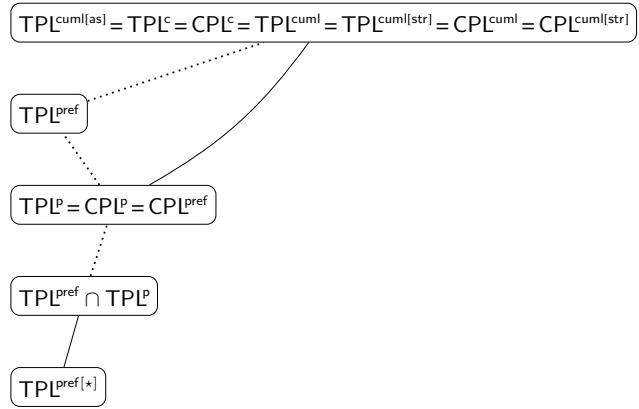
\begin{center}
    \resizebox{0.99\columnwidth}{!}{\begin{tikzpicture}
            \node[anchor=west,minimum height=0.6cm,draw,rectangle,rounded corners] (TPLPrefStar) at (0,0) {\( \pmLogicX{\TPL}{\logicFont{pref}\ensuremath{[\star]}} \)};
            \node[anchor=west,minimum height=0.6cm,draw,rectangle,rounded corners] (TPLPrefCapP) at ([yshift=4em]TPLPrefStar.west) {\( \TPLPref \cap \pmLogicX{\TPL}{\logicFont{p}} \)};
            
            \node[anchor=west,minimum height=0.6cm,draw,rectangle,rounded corners] (TPLp) at ([yshift=4em]TPLPrefCapP.west) {\( \TPLSystemP\)\,{=}\,\( \CPLSystemP \)\,{=}\,\( \CPLPref \)};
            \node[anchor=west,minimum height=0.6cm,draw,rectangle,rounded corners] (TPLPref) at ([yshift=4em]TPLp.west) {\( \TPLPref \)};
            \node[anchor=west,minimum height=0.6cm,draw,rectangle,rounded corners] (TPLc) at ([yshift=4em]TPLPref.west) {\( \TPLAsym \)\,{=}\,\( \TPLSystemC \)\,{=}\,\( \CPLSystemC \)\,{=}\,\( \TPLCum \)\,{=}\,\( \TPLStrongCum \)\,{=}\,\( \CPLCum \)\,{=}\,\( \CPLStrongCum \)};
            
            \draw (TPLPrefStar) edge [] (TPLPrefCapP);
            \draw (TPLPrefCapP) edge [dotted,thick] (TPLp);
            \draw (TPLp) edge [dotted,thick] (TPLPref);
            \draw (TPLp) edge [bend right=12] (TPLc);
            \draw (TPLPref) edge [dotted,thick] (TPLc);
    \end{tikzpicture}}
\end{center}
Next, we consider the hierarchy of classes of entailment relations over the language \( \pdl\):\\
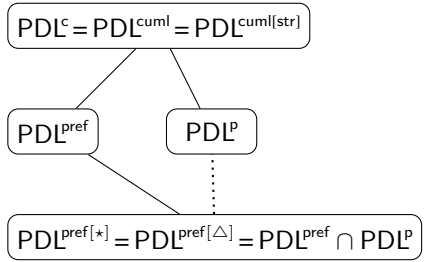
\begin{center}
    \begin{tikzpicture}                
        \node[anchor=west,minimum height=0.6cm,draw,rectangle,rounded corners] (PDLPrefStar) at (0,0) {\( \pmLogicX{\PDL}{\logicFont{pref}\ensuremath{[\star]}} \)\,{=}\,\( \pmLogicX{\PDL}{\logicFont{pref}\ensuremath{[\triangle]}} \)\,{=}\,\( \PDLPref \cap \PDLSystemP \)};
        \node[anchor=west,minimum height=0.6cm,draw,rectangle,rounded corners,minimum width=1.2cm] (PDLp) at ([xshift=6em,yshift=4em]PDLPrefStar.west) {\( \PDLSystemP \)};
        
        \node[anchor=west,minimum height=0.6cm,draw,rectangle,rounded corners,minimum width=1.2cm] (PDLPref) at ([yshift=4em]PDLPrefStar.west) {\( \PDLPref \)};
        \node[anchor=west,minimum height=0.6cm,draw,rectangle,rounded corners] (PDLc) at ([yshift=4em]PDLPref.west) {\( \PDLSystemC \)\,{=}\,\( \PDLCum \)\,{=}\,\( \PDLStrongCum \)};

        \draw (PDLPrefStar) edge [dotted,thick] (PDLp);
        \draw (PDLPrefStar) edge [] (PDLPref);
        \draw (PDLp) edge [] (PDLc);
        \draw (PDLPref) edge [] (PDLc);
    \end{tikzpicture}
\end{center}
Classes not defined in the main part of this paper are defined in SMK~({\protect\citeyear{KS_SauerwaldMeierKontinen2025}}) and KLM~({\protect\citeyear{KS_KrausLehmannMagidor1990}}).
We briefly describe them below; for the exact definitions, please consult the respective literature.
\begin{itemize}
    \item \( \CPLPref \) contains all entailment relations for \( \CPL \) that are defined by a preferential model~\cite{KS_KrausLehmannMagidor1990}, i.e., cumulative models in which for every state \( S\), the set $\pmL(s)$ contains exactly one element, and $\pmR$ is a strict partial order.
    \item \( \PDLPref \), respectively \( \TPLPref \), from SMK (\citeyear{KS_SauerwaldMeierKontinen2025}) contain all entailment relations for \( \PDL \), respectively \( \TPL\), that are defined for preferential models (see at \( \CPLPref \)). 
    
    \item \( \PDLSystemP \) and \( \TPLSystemP \) from SMK (\citeyear{KS_SauerwaldMeierKontinen2025}), and \( \CPLSystemP \) from KLM~({\protect\citeyear{KS_KrausLehmannMagidor1990}}), are analogously defined to \( \PDLSystemC \), \( \TPLSystemC \) and \( \CPLSystemC \) via using \textit{System~P} instead of {\SystemC}.
    One obtains \textit{System~P} by extending {\SystemC} by the following rule:
    \begin{align}	&\frac{\formulaOne\nmableit\formulaThree\hspace{0.5cm}\formulaTwo\nmableit\formulaThree}{\formulaOne\lor\formulaTwo\nmableit\formulaThree} \tag{Or}\label{pstl:Or}
    \end{align}
    \item \( \pmLogicX{\PDL}{\logicFont{pref}\ensuremath{[\star]}} \), respectively \( \pmLogicX{\TPL}{\logicFont{pref}\ensuremath{[\star]}} \), from SMK (\citeyear{KS_SauerwaldMeierKontinen2025}) is the restrictions to preferential models which satisfy the so-called \eqref{eq:StarProperty}-property,
    \begin{equation}
        \minOf{\modelsOf{\formulaOne \lor \formulaTwo}}{\pmP} \subseteq \minOf{\modelsOf{\formulaOne}}{\pmP} \cup \minOf{\modelsOf{\formulaTwo}}{\pmP}	\tag{\( \star \)}\label{eq:StarProperty}
    \end{equation}
    forcing that minimal models of a disjunction split to minimal models of the disjuncts.
    \item \( \pmLogicX{\PDL}{\logicFont{pref}\ensuremath{[\triangle]}} \) from SMK (\citeyear{KS_SauerwaldMeierKontinen2025}) is \( \PDLPref \) over preferential models restricted such that minimal states contain only singleton teams.
\end{itemize}

 \section{Complete Proofs}

\newenvironment{repeatprop}[1]{\smallskip\par\noindent\textbf{Proposition~\ref{#1}}.\itshape}{\par}
\newenvironment{repeatthm}[1]{\smallskip\par\noindent\textbf{Theorem~\ref{#1}}.\itshape}{\par}
\newenvironment{repeatcorollary}[1]{\smallskip\par\noindent\textbf{Corollary~\ref{#1}}.\itshape}{\par}
\newenvironment{repeatobs}[1]{\smallskip\par\noindent\textbf{Observation~\ref{#1}}.\itshape}{\par}
\newenvironment{repeatlemma}[1]{\smallskip\par\noindent\textbf{Lemma~\ref{#1}}.\itshape}{\par}

In this supplement, we provide more detailed and missing proofs from the main paper.

\begin{repeatprop}{prop:systemCteams}
	\( \PDLCum \subseteq  \PDLSystemC \).
\end{repeatprop}
\begin{proof}
	Let \( {\nmableitC} \in \PDLCum \) be a cumulative entailment relation that is based on a cumulative model \( \pmC=\tuple{\pmS,\pmL,\pmR} \).
	We show that \( {\nmableitC} \) satisfies all rules of {\SystemC}:\\
		\noindent[\emph{\ref{pstl:Ref}.}]
		Considering the definition of \( \nmableitC \) yields that \( \formulaOne \nmableitC \formulaOne \) if for all minimal \( s\in\statesOf{\formulaOne} \) it holds that \( \ell(s)\models\formulaOne \). By the definition of \( \statesOf{\formulaOne} \), we have \( s\in\statesOf{\formulaOne} \) if \( \ell(s)\models\formulaOne \). Consequently, we have that \( \formulaOne\nmableitC \formulaOne \).

		\noindent[\emph{\ref{pstl:LLE}.}] 
		From \( \formulaOne\equiv\formulaTwo \), we obtain that \( \statesOf{\formulaOne}=\statesOf{\formulaTwo} \) holds.
		By using this last observation and the definition of \( \nmableitC \), we obtain \( \formulaTwo\nmableitC\formulaThree \) from \( \formulaOne\nmableitC\formulaThree \).

		\noindent[\emph{\ref{pstl:RW}.}] 
		Clearly, by definition of \( \formulaOne\models\formulaTwo \) we have that \( \modelsOf{\formulaOne}^\pmC\subseteq\modelsOf{\formulaTwo}^\pmC \).
		From the definition of \( \formulaThree\nmableitC\formulaOne \), we obtain that \( \ell(s)\models\formulaOne \) holds for each minimal \( s\in\statesOf{\formulaThree} \).
		The condition \( \ell(s)\models\formulaOne \) in the last statement is equivalent to stating \( \ell(s)\in\modelsOf{\formulaOne}^\pmC \). Because of \( \modelsOf{\formulaOne}^\pmC\subseteq\modelsOf{\formulaTwo}^\pmC \), we also have that \( \ell(s)\in\modelsOf{\formulaTwo}^\pmC \); and hence, \( \ell(s)\models\formulaTwo \)  for each minimal  \( s\in\statesOf{\formulaThree} \).
		This shows that \( \formulaThree\nmableitC\formulaTwo \) holds.

		\noindent[\emph{\ref{pstl:Cut}.}] 		
		By unfolding the definition of \( \nmableitC \), we obtain 
		\(  {\minOf{\statesOf{\formulaOne\land\formulaTwo}}{\pmR}}  \subseteq \statesOf{\formulaThree} \) 
		from \( \formulaOne\land\formulaTwo \nmableitC \formulaThree \).
		Analogously, \( \formulaOne \nmableitC \formulaTwo \) unfolds to 
		\( \minOf{\statesOf{\formulaOne}}{\pmR} \subseteq \statesOf{\formulaTwo} \).
		Moreover, employing basic set theory yields that \( \statesOf{\formulaOne\land\formulaTwo}=\statesOf{\formulaOne}\cap\statesOf{\formulaTwo}\subseteq\statesOf{\formulaOne} \) holds.
		From \( \statesOf{\formulaOne\land\formulaTwo}\subseteq\statesOf{\formulaOne} \) and \( \minOf{\statesOf{\formulaOne}}{\pmR} \subseteq \statesOf{\formulaTwo} \), we obtain \( \minOf{\statesOf{\formulaOne}}{\pmR} \subseteq \statesOf{\formulaOne\land\formulaTwo} \).
		Consequently, we also have that \( \minOf{\statesOf{\formulaOne}}{\pmR}=\minOf{\statesOf{\formulaOne\land\formulaTwo}}{\pmR} \) holds.
		Using the last observation and \( \minOf{\statesOf{\formulaOne\land\formulaTwo}}{\pmR}  \subseteq \statesOf{\formulaThree} \), we obtain \( \minOf{\statesOf{\formulaOne}}{\pmR}  \subseteq \statesOf{\formulaThree} \). Hence also \( \formulaOne \nmableitC \formulaThree \) holds.

		\noindent[\emph{\ref{pstl:CM}.}] 		
		By unfolding the definition of \( \nmableitC \), we obtain \(  \minOf{\statesOf{\formulaOne}}{\pmR}  \subseteq \statesOf{\formulaTwo} \)  and \(  \minOf{\statesOf{\formulaOne}}{\pmR}  \subseteq \statesOf{\formulaThree} \).
		We have to show that \( \minOf{\statesOf{\formulaOne\land\formulaTwo}}{\pmR} \subseteq \statesOf{\formulaThree} \) holds.
		If \( \statesOf{\formulaOne\land\formulaTwo}=\emptyset \) holds, we immediately obtain \( \minOf{\statesOf{\formulaOne\land\formulaTwo}}{\pmR} = \emptyset \subseteq \statesOf{\formulaThree} \) holds. We continue with the case of \( \statesOf{\formulaOne\land\formulaTwo} \neq \emptyset \).
		Because \( \pmR \) is smooth, we have \( \minOf{\statesOf{\formulaOne\land\formulaTwo}}{\pmR} \neq \emptyset \).
		Let \( s \) be element of \( \minOf{\statesOf{\formulaOne\land\formulaTwo}}{\pmR} \).
		Clearly, we have that \( s \in \statesOf{\formulaOne}  \) holds. We show by contradiction that \( s \) is minimal in \( \statesOf{\formulaOne} \). 
		Assume that \( s \) is not minimal in \( \statesOf{\formulaOne} \).
		From the smoothness condition, we obtain that there is a state \( s'\in \statesOf{\formulaOne} \) such that \( s' \pmR s \) and \( s' \) is minimal in \( \statesOf{\formulaOne} \) with respect to \( \pmR \). 
		Because \( s' \) is minimal in \( \statesOf{\formulaOne} \)  and because we have \( \minOf{\statesOf{\formulaOne}}{\pmR}  \subseteq \statesOf{\formulaTwo} \), we also have that \(  s'\in \statesOf{\formulaTwo} \) holds and hence that \( s'\in \statesOf{\formulaOne\land\formulaTwo} \) holds. The latter contradicts the minimality of \( s \) in \( \statesOf{\formulaOne\land\formulaTwo} \).
		Consequently, we have that \( s \in \minOf{\statesOf{\formulaOne}}{\pmR} \) holds. 
		Because we have that \(  \minOf{\statesOf{\formulaOne}}{\pmR}  \subseteq \statesOf{\formulaThree} \), we obtain \( \formulaOne \land \formulaTwo \nmableitC \formulaThree \).\qedhere
	\end{proof}

Next, we will consider Lemma~\ref{lem:closures},
\begin{repeatlemma}{lem:closures}
    The following statements are true:
	\begin{enumerate}[(a)]
		\item 
		For \( {\nmableit} \in \PDLSystemC \) we have that \( \Cableit(\formulaOne) = \Cn(\Cableit(\formulaOne)) \).
		\item 
		For each formula \( \formulaOne \) there exists an \( M \subseteq \mathcal{P}(\mathbb{T}_{N}) \) such that \( \Cn(\formulaOne) = \Th(M) \) and \( M = \modelsOf{\formulaOne} \).
		\item 
		For each set of formulas \( F \) with \( F=\Cn(F) \) there exists a formula \( \formulaOne_{F} \) such that \( F = \Cn(\formulaOne_{F}) \) is true.
	\end{enumerate}
\end{repeatlemma}
The following Lemma~\ref{lem:RWclosure} to Lemma~\ref{lem:formfinite} prove Lemma~\ref{lem:closures}.
We start with observing that consequences under {\SystemC} are closed under classical consequence.
	\begin{lemma}
		\label{lem:RWclosure}
		For \( {\nmableit} \in \PDLSystemC \) we have that \( \Cableit(\formulaOne) = \Cn(\Cableit(\formulaOne)) \).
	\end{lemma}
	\begin{proof}
		We obtain \( \Cableit(\formulaOne) \subseteq \Cn(\Cableit(\formulaOne)) \) by definition and reflexivity of \( \models \).
		Now assume that \(  \Cableit(\formulaOne) \subsetneq \Cn(\Cableit(\formulaOne)) \) is true, i.e., there is \( \formulaTwo \in \Cableit(\formulaOne) \) and \( \formulaThree \in \Cn(\Cableit(\formulaOne)) \) such that \( \formulaTwo \models \formulaThree \) and \( \formulaThree \notin \Cableit(\formulaOne) \).
		Now recall that \( {\nmableit} \) satisfies \eqref{pstl:RW}, and hence, we also have \( \formulaOne \nmableit \formulaThree \), because \( \formulaTwo \models \formulaThree \) and \( \formulaOne \nmableit \formulaTwo \) is true. From \( \formulaOne \nmableit \formulaThree \), we obtain the contradiction \( \formulaThree \in \Cableit(\formulaOne) \).
	\end{proof}
	
	The theory \( \Th(M) \) of a set of teams \( M \) is given by
	\begin{equation*}
		\Th(M) = \{\, \formulaOne \in \ssFormulas \mid X \models \formulaOne  \text{ for all } X \in M \,\} \ .
	\end{equation*}
	For every set of classical consequences of a formula \( \varphi \), there is a set of teams whose theory yields the same consequences.
    Because this is a direct consequence of the semantics, we omit a proof here.
	\begin{lemma}
		\label{lem:formtheorie}
		For each formula \( \formulaOne \) there exists an \( M \subseteq \mathcal{P}(\mathbb{T}_{N}) \) such that \( \Cn(\formulaOne) = \Th(M) \) and \( M = \modelsOf{\formulaOne}^t \).
	\end{lemma}
	Because we consider online finite signatures here, we obtain the following result.
	
	\begin{lemma}
		\label{lem:formfinite}
		For each set of formulas \( F \) with \( F=\Cn(F) \) there exists a formula \( \formulaOne_{F} \) such that \( F = \Cn(\formulaOne_{F}) \).
	\end{lemma}
	\begin{proof}
		We will employ the underlying logic \( \PDL_{N} \), where we have a finite signature \( N \subseteq \Prop \).
		It has been shown that classical disjunction is expressible  in  \( \PDL_{N} \), i.e., for two formulas \( \formulaOne,\formulaTwo \) there is a formula \( \formulaOne \ovee \formulaTwo \) with \( \modelsOf{\formulaOne \ovee \formulaTwo}^t = \modelsOf{\formulaOne}^t\cup \modelsOf{\formulaTwo}^t \) 
		(\citeauthor{Yang_dissertation} \citeyear{Yang_dissertation}).
		Moreover, every downward-closed set of teams is \( M \) definable, i.e., there is a \( \formulaOne_{M} \) with \( \modelsOf{\formulaOne_{M}}^t = M \).
		Because \( N \)  is finite, for every set of formulas \( F \) the set \( \modelsOf{F}^t \) is a finite disjunction of downward-closed sets of teams, i.e., \( \modelsOf{F}^t = M_1\cup \ldots \cup M_n \).
		Hence, for \(  \formulaOne_{F} = \ovee_{i=1}^{n}  \formulaOne_{M_i} \) we have \( \modelsOf{F}^t = \modelsOf{\formulaOne_{F}}^t \).
	\end{proof}

\end{document}